\documentclass{ecai}  

\usepackage{graphicx}
\usepackage{latexsym}
\usepackage{hhline,multirow}%
\usepackage{enumerate}
\usepackage[inline]{enumitem}
\usepackage{amsmath,amssymb,amsfonts}%
\usepackage{algorithm}%
\usepackage{float}
\usepackage{algorithmic}%

\makeatletter
\DeclareRobustCommand{\qed}{%
  \ifmmode 
  \else \leavevmode\unskip\penalty9999 \hbox{}\nobreak\hfill
  \fi
  \quad\hbox{\qedsymbol}}
\newcommand{\openbox}{\leavevmode
  \hbox to.77778em{%
  \hfil\vrule
  \vbox to.675em{\hrule width.6em\vfil\hrule}%
  \vrule\hfil}}
\newcommand{\qedsymbol}{\openbox}
\newenvironment{proof}[1][\proofname]{\par
  \normalfont
  \topsep6\p@\@plus6\p@ \trivlist
  \item[\hskip\labelsep\itshape
    #1.]\ignorespaces
}{%
  \qed\endtrivlist
}
\newcommand{\proofname}{Proof}
\makeatother

\usepackage{booktabs}%

\def \MB{\mathbb{M}_{BFS}}
\def \MD{\mathbb{M}_{DFS}}
\def \MG{\mathbb{M}_{AP}}

\def \AA{\mathcal{A}}

\def \qq{\boldsymbol{q}}
\def \bb{\boldsymbol{b}}

\def \safepolicyp{FCFS policy}
\def \safepoliciesp{FCFS policies}
\def \safepolicy{\safepolicyp }
\def \safepolicies{\safepoliciesp }

\def \SP{FCFSP}

\newtheorem{lemma}{Lemma}
\newtheorem{thm}{Theorem}
\newtheorem{example}{Example}

\begin{document}

\begin{frontmatter}

\title{On the Manipulability of Maximum Vertex-Weighted Bipartite $b$-matching Mechanisms}

\author[A]{\fnms{Gennaro}~\snm{Auricchio}\thanks{Corresponding Author. Email: ga647@bath.ac.uk}}
\author[A]{\fnms{Jie}~\snm{Zhang}}

\address[A]{Department of Computer Science, University of Bath}

\begin{abstract}
In this paper, we study the Maximum Vertex-weighted $b$-Matching (MVbM) problem on bipartite graphs in a new game-theoretical environment. In contrast to other game-theoretical settings, we consider the case in which the value of the tasks is public and common to every agent so that the private information of every agent consists of edges connecting them to the set of tasks. In this framework, we study three mechanisms. Two of these mechanisms, namely $\MB$ and $\MD$, are optimal but not truthful, while the third one, $\MG$, is truthful but sub-optimal. Albeit these mechanisms are induced by known algorithms, we show $\MB$ and $\MD$ are the best possible mechanisms in terms of Price of Anarchy and Price of Stability, while $\MG$ is the best truthful mechanism in terms of approximated ratio. Furthermore, we characterize the Nash Equilibria of $\MB$ and $\MD$ and retrieve sets of conditions under which $\MB$ acts as a truthful mechanism, which highlights the differences between $\MB$ and $\MD$. Finally, we extend our results to the case in which agents' capacity is part of their private information.
%
\end{abstract}

\end{frontmatter}

\section{Introduction}

Managers of large companies are periodically required to submit a list of the projects they carried out for external evaluations.
Alongside the list of projects, a manager needs to identify a worker liable for the achievement of every project.
Due to the workload cap, every worker can be found responsible for only a finite number of projects.
Moreover, every project has a different prestige, which can be regarded as its overall value. 
The manager aims to maximize the total value of the reported projects by deploying its staff following the aforementioned constraints.
Meanwhile, every worker is interested in being associated with high-value projects rather than low-value ones. 
Therefore, there might be workers that benefit by hiding part of their connections to low value projects.
A similar situation occurs in the universities of several European countries.
Indeed, to assess the research impact of their higher education institutions, the country asks the universities for a list of their best publications, along with the name of the main author.\footnote{This, for example, is a common practice in the United Kingdom, see \url{https://www.ref.ac.uk} for a reference.}
On one hand, every university wants to find a matching between its lecturers and the publications that maximize the total impact.
On the other hand, individual academics want to be designated as the main author of their best publications.
In both scenarios, we need to allocate a set of resources that does possess an objective and publicly known value, be they projects and their prestige or publications and their impact score, to a set of self-interested agents.
Aside from these two examples, there are many other real-life situations that can be rephrased as matching problems with self-interested agents.
Matching problems were first introduced to minimize transportation costs \cite{Hitchcock1941,Kantorovitch1958} and to optimally assign workers to job positions \cite{Easterfield1946,Thorndike1950}. 
Thereafter, bipartite matching found application in several applied problems, such as sponsored searches \cite{DBLP:journals/sigact/BirnbaumM08,DBLP:journals/fttcs/Mehta13}, school admissions \cite{abdulkadirouglu2005college,DBLP:journals/tcs/BiroFIM10}, scheduling \cite{ZHAO2010837,DBLP:conf/soda/GuptaKS14}, and general resource allocation \cite{feng2013device,7295474}.
Indeed, characterizing matching through the maximization or minimization of a functional allows to define mathematical objects that arise in several applied contexts. 
For example, Maximum Cardinality Matchings have connections with the computation of perfect matchings \cite{fukuda1994finding}, bottleneck matchings \cite{edmonds1970bottleneck}, and Lévy-Prokhorov distances, \cite{lahn2021faster}.
Another example is the Maximum Edge-weighted Matching problem, which has been widely used in clustering problems \cite{ding2019clustering}, machine learning \cite{bayati2005maximum}, and also to compute Wasserstein-barycenters \cite{auricchio2019computing}.
For a complete discussion of the matching problems and their applications, we refer to \cite{lovasz2009matching}.

\paragraph*{Game-theoretic aspects of matching problems.}
Matching problems have been extensively studied from a game-theoretical perspective in various contexts, including one-sided matching (such as house allocation \cite{HZ79,10.1257/000282802762024728} and the resident/hospital problem \cite{manlove2013algorithmics}) and two-sided matching (such as stable marriage \cite{GaleS13}, student admission \cite{balinski1999tale}, PhD grants assignment \cite{cechlarova2019problem}, and market matching \cite{kojima2019new}). 
Numerous variants of these basic models have been proposed, often with meaningful constraints such as regional constraints \cite{kamada2017recent,goto2016strategyproof}, diversity constraints \cite{marcolino2014give,10.5555/3306127.3331717,stablematchdiversityNP,ehlers2014school}, or lower quotas \cite{boehmer2022fine,yokoi2017generalized}. 
Perhaps the most generic variation on this framework has been studied in \cite{fadaei2017truthful} and \cite{azar2015truthful}, where the authors adapt the game-theoretical framework to the class of Generalized Assignment Problems.
In every setting containing self-interested agents, it is important to assess the impact that the mechanism has on the social problem.
Consequently, several papers have examined the social aspects of these mechanisms, including manipulability \cite{fragiadakis2016strategyproof,krysta2014size,chakrabarty2014welfare}, fairness \cite{kamada2023fair}, and envy-freeness \cite{abdulkadiroǧlu2020efficiency,abdulkadirouglu2003school}.
For a comprehensive survey of the game-theoretical properties of mechanisms in matching problems, we refer the reader to \cite{aziz2022matching}.
All these works, however, assume that the values of the tasks are subjective, meaning that the same task may have a different value for different agents.
Yet, in many cases, the value of the objects to be assigned is publicly known. 
Thus, we cannot model the agents' private information as their ordinal or cardinal preferences over the set of tasks. 
In this paper, we focus on this distinct class of problems.

\begin{table}[t]
\begin{center}
\caption{The mechanisms we study and their properties in the two game-theoretical settings we consider . The "Yes*" indicates that the mechanism is group strategyproof under minor assumptions. \label{table:schedule}}
    \begin{tabular}{c | c | c | c | c }
 \midrule[1.5pt]
    & Truthful & Group SP  & Optimal  & Efficiency \\
 \hline
 $\MB$   & No & No & Yes & PoA = PoS = 2  \\
   $\MD$  & No & No & Yes & PoA = PoS = 2 \\
 $\MG$   & Yes & Yes* & No & $ar$ = 2 \\
\midrule[1.5pt]
\end{tabular}

\end{center}
\end{table}

\paragraph*{Our Contributions.}
In this paper, we consider a game-theoretical framework in which the agents' private information is the set of edges connecting the agent to the tasks.
We assume that the values of the tasks are public and common to every agent.
Moreover, we assume that the agents are bounded by their statements so that agents can hide edges, but they cannot report non-existing edges.
In this framework, we study three mechanisms induced by known algorithms.
Two of them, namely $\MB$ and $\MD$, are induced by the algorithm proposed in \cite{spencer1984node}.
The third one, namely $\MG$, is induced by the algorithm proposed in \cite{dobrian20192}, which is an approximation of the one proposed in \cite{spencer1984node}.
First, we study the truthfulness and the group strategyproofness of these mechanisms.
Both $\MB$ and $\MD$ are optimal, but not truthful nor group strategyproof.
On the contrary, $\MG$ is not optimal, but truthful and, under mild assumptions, also group strageyproof.
We show that $\MG$ has an approximation ratio (ar) equal to $2$ and that this value is tight, i.e. no other deterministic and truthful mechanism can beat this constant value.
We then study the Nash Equilibrium induced by $\MB$ and $\MD$, we use these characterizations to prove that the Price of Anarchy (PoA) and the Price of Stability (PoS) of both mechanisms is equal to $2$.
As for $\MG$, these efficiency guarantees are tight: there does not exist a mechanism that can achieve a smaller PoA or PoS.
We then characterize some classes of inputs on which $\MB$ is truthful, while $\MD$ is not.
In particular, we infer that, despite their similarities, $\MB$ and $\MD$ do possess different game-theoretical properties.
Finally, we extend our results to the case in which the agents are able to report their capacity along with their edges.
In Table \ref{table:schedule}, we summarize our main findings.

\section{Preliminaries}
\label{sect:prelim}

In this section, we recall the basic notions on the MVbM problem and introduce the algorithm defining the mechanisms we study.

{\bf The Maximum Vertex-weighted $b$-Matching Problem.}
Let $G=(A \cup T,E)$ be a bipartite graph. 
Throughout the paper we refer to $A=\{a_1,a_2, \cdots, a_{n}\}$ as the set of agents and refer to $T=\{t_1,t_2, \cdots, t_{m}\}$ as the set of tasks.
The set $E$ contains the edges of the bipartite graph.
We say that an edge $e\in E$ belongs to agent $a_i$ if $e=(a_i,t_j)$ for a $t_j\in T$.
We denote with $M=|E|$ and $N=|A \cup T|=n+m$ the number of edges and the total number of vertices of the graph, respectively. 
Since the graphs are undirected, we denote an edge by $(a_i,t_j)$ and $(t_j,a_i)$ interchangeably. 
Moreover, for any given agent $a_i\in A$, we define the set $T_{E,i}$ as the set of tasks in $T$ that are connected to agent $a_i$ through the set of edges $E$.
When it is clear from the context which set $E$ we are referring to, we drop the subscript from $T_{E,i}$ and use $T_i$.
Let $\bb=(b_1,b_2,\dots,b_n)$ be the vector containing the capacities of the agents, where $b_i\in \mathbb{N}$ for every $i=1,\dots,n$.
A subset $\mu\subset E$ is a $b$-matching if, for every vertex $a_{i}\in A$, the number of edges in $\mu$ linked to $a_i$ is less than or equal to $b_i$ and, for every vertex $t_j\in T$, the number of edges in $\mu$ linked to $t_j$ is, at most, one.
Given a $b$-matching $\mu$, the vertex $a_i \in A$ is \emph{saturated} with respect to $\mu$ if the number of edges in $\mu$ linked to $a_i$ is exactly $b_i$.
Otherwise, the vertex $a_i$ is \emph{unsaturated} with respect to $\mu$. 
We denote with $\qq=(q_1,q_2,\dots,q_m)$ the vector containing the values of the tasks, in our framework, we have $q_j>0$ for every $j$. 
The value of a matching $\mu$ is $w(\mu):=\sum_{t_j\in T_{\mu}}q_{j}$, where $T_{\mu}\subset T$ is the set of tasks matched by $\mu$.
Given a bipartite graph and a vector $\bb$, the MVbM problem consists in finding a $b$-matching $\mu$ that maximizes  $w(\mu)$. 
Finally, given two sets of edges $\mu_1$ and $\mu_2$, denote with $\mu_1\oplus\mu_2=(\mu_1\backslash\mu_2)\cup(\mu_2\backslash\mu_1)$ their \emph{symmetric difference}.
A path $P = \{ (t_{j_1},a_{i_1}), (a_{i_1},t_{j_2}), (t_{j_2},a_{i_2}), \cdots , (t_{j_L}, a_{i_L}) \}$ in $G$ is a sequence of edges that joins a sequence of vertices. 
We say that $P$ has a length equal to $\lambda$ if it contains $\lambda$ edges.
Given a path $P$ and a $b$-matching $\mu$, $P$ is an \emph{augmenting path} with respect to $\mu$ if every vertex $a_{i_\ell}$, for $\ell=1,\dots,L-1$, is saturated with respect to $\mu$, $a_{i_L}\in A$ is unsaturated with respect to $\mu$, and the edges in the path alternatively do not belong to $\mu$ and belong to $\mu$. 
That is, $(t_{j_\ell},a_{i_\ell})\notin \mu$,
$\ell\in[L]$ and $(a_{i_\ell}, t_{j_{\ell+1}})\in \mu$, $\ell\in[L-1]$, where $[L]$ is the set containing the first $L$ natural numbers, i.e. $[L]:=\{1,2,\dots,L\}$.

{\bf The Algorithm Outline.}
In this paper, we consider two well-known algorithms from a mechanism design perspective.
The first algorithm is presented in \cite{spencer1984node} and its routine consists in defining a sequence of matchings, namely $\{\mu_j\}_{j=0,1,\dots,m}$, of increasing cardinality.
Indeed, the first matching is set as $\mu_0=\emptyset$ and, given $\mu_{j-1}$, $\mu_j$ is defined as follows: if there exists $P_j$ an augmenting path (with respect to $\mu_{j-1}$) that starts from $t_j$, then $\mu_j=\mu_{j-1}\oplus P_j$, otherwise $\mu_j=\mu_{j-1}$. 
In Algorithm \ref{alg:Maxvalue}, we sketch the routine of the algorithm.
In \cite{spencer1984node} it has been shown that Algorithm \ref{alg:Maxvalue} returns a solution to the MVbM problem in $O(NM)$ time, regardless of whether we find the augmenting path using the \emph{Breadth-First Search} (BFS) or the \emph{Depth-First Search} (DFS). 
Both the DFS and the BFS when they traverse the graph in search of an augmenting path implicitly assume that there is an ordering of the agents.
We say that agent $a_i$ has a higher \emph{priority} than agent $a_j$ if $a_i$ is explored before $a_j$.
In particular, $a_1$ has the highest priority.
We say that Algorithm \ref{alg:Maxvalue} is endowed with the BFS (or endowed with the DFS) if we use the BFS (or DFS) to find an augmenting path.
%
%

%
The second algorithm is introduced in \cite{dobrian20192} and \cite{al20222} and is an approximation version of Algorithm \ref{alg:Maxvalue} that searches only among the augmenting paths whose length is $1$.
The authors showed that this approximation algorithm finds a matching whose weight is, in the worst case, half the weight of the MVbM.
Notice that for this approximation algorithm, using the BFS or the DFS does not change the outcome thus we omit which traversing graph method is used.

\begin{algorithm}[tb]
\caption{Algorithm for MVbM, \cite{spencer1984node}}
\label{alg:Maxvalue}
\textbf{Input}: A bipartite graph $G=(A \cup T,E)$; agents' capacities $b_{i}, i=1,\cdots,n$; task weights $q_j, j=1,\cdots,m$.\\
\textbf{Output}:  An MVbM $\mu$.

\begin{algorithmic}[1] 
\STATE $\mu_0\leftarrow\emptyset$;
\STATE Sort $q_j$, $j=1,\cdots,m$, in decreasing order;
\FOR{each $j\in T$}
\IF {there is augmenting path $P$ starting from $j$ w.r.t. $\mu_{j-1}$}
\STATE $\mu_j=\mu_{j-1}\oplus P$;
\ELSE
\STATE $\mu_{j}=\mu_{j-1}$;
\ENDIF
\ENDFOR
\STATE \textbf{return} $\mu=\mu_{m}$;
\end{algorithmic}
\end{algorithm}

\section{The Game-Theoretical framework }
\label{subsec:gtfwagent}

%
%
%

%
In this section, we formally define the space of the agents' strategies and the mechanisms we study throughout the paper.
%
%

{\bf The Strategy Space of the Agents.}
Given a bipartite graph $G=(A\cup T, E)$, a capacity vector $\bb$, a value vector $\qq$, and a $b$-matching $\mu$ over $G$, we define the social welfare achieved by $\mu$ as the total value of the tasks assigned to the agents.
Since the social welfare achieved by $\mu$ is equal to the total weight of the matching, we use $w(\mu)$ to denote it.
We then define the utility of agent $a_i$ as $w_i(\mu):=\sum_{t_j\in T_{\mu,i}} q_{j}$, 
where $T_{\mu,i}$ is the set of tasks that $\mu$ matches with agent $a_i$.
Notice that $w(\mu)=\sum_{i=1}^{n}w_i(\mu)$.
In this paper, we focus on two settings: 
\begin{enumerate*}[label=(\roman*)]
    \item The private information of each agent consists of the set of edges that connect the agent to the tasks' set. We call this setting Edge Manipulation Setting (EMS).
    \item The private information of each agent consists of the set of edges and its own capacity. We call this setting Edge and Capacity Manipulation Setting (ECMS).
\end{enumerate*}
In both cases, we assume that every agent is bounded by its statement, thus every agent is able to report incomplete information, but they are not allowed to report false information.
In EMS, this means that an agent can hide some of the edges that connect it to the tasks, but it cannot report an edge that does not exist.
The set of strategies of agent $a_i$, namely $\mathcal{S}_i$, is, therefore, the set of all the possible non-empty subsets of $T_i$, where $T_i$ is the set containing all the existing edges that connect $a_i$ to the tasks.
In ECMS, this means that an agent can report only a capacity that is lower than its real one and that it cannot report an edge that does not exist.
In this case, the set of strategies of agent $a_i$ is then $\mathcal{S}_i\times [b_i]$, where $b_i$ is the real agent's capacity.
%

%
For the sake of simplicity, \textbf{from now on all the results and definitions we introduce are for the EMS}, unless we specify otherwise.
We generalize our results to the ECMS in Section \ref{subsect:ECMS}.
{\bf The Mechanisms.}
A mechanism for the MVbM problem is a function $\mathbb{M}$ that takes as input the private information of the agents and returns a $\bb$-matching.
We denote with $\mathcal{I}_{\mathbb{M}}$ the set of  possible inputs for $\mathbb{M}$ in the EMS.
In our paper, we consider three mechanisms:
\begin{enumerate}[label=(\roman*)]
    \item $\MB$, which takes in input the edges of the agents and uses Algorithm \ref{alg:Maxvalue} endowed with the BFS to select a matching\footnote{Throughout the paper, we use $\MB$ to denote both the mechanism that takes as input the edges and both the edges and the capacity of every agent. It will be clear from the context which is the input of the mechanism. The same goes for the other two mechanisms.}.
    \item $\MD$, which takes in input the edges of the agents and uses Algorithm \ref{alg:Maxvalue} endowed with the DFS to select a matching.
    \item $\MG$, which takes in input the edges of the agents and uses the approximated version of Algorithm \ref{alg:Maxvalue} to select a matching.
\end{enumerate}
Notice that $\MB$ and $\MD$ are optimal since they are both induced by Algorithm \ref{alg:Maxvalue}.
Given a mechanism $\mathbb{M}$, every element $I\in \mathcal{I}_{\mathbb{M}}$ is composed by the reports of $n$ self-interested agents, so that $\mathcal{I}_{\mathbb{M}}=\otimes_{i=1}^n\mathcal{I}_{i}$, where $\mathcal{I}_i$ is the set of feasible inputs for agent $a_i$.
We say that a mechanism $\mathbb{M}$ is \textit{strategyproof} (or, equivalently, \textit{truthful})  for the EMS if no agent can get a higher payoff by hiding edges.
More formally, if $I_i$ is the true type of agent $a_i$, it holds true that 
\[
w_i(\mathbb{M}(I'_i,I_{-i}))\le w_i(\mathbb{M}(I_i,I_{-i}))
\]
for every $I'_i\in \mathcal{S}_i$.
Another important property for mechanisms is the group strategyproofness.
A mechanism is \textit{group strategyproof} for agent manipulations if no group of agents can collude to hide some of their edges in such a way that
\begin{enumerate*}[label=(\roman*)]
    \item the utility obtained by every agent of the group after hiding the edges is greater than or equal to the one they get by reporting truthfully,
    \item at least one agent gets a better payoff after the group hides the edges.
\end{enumerate*}

To evaluate the performances of the mechanisms, we use the Price of Anarchy (PoA), Price of Stability (PoS), and approximation ratio (ar), which we briefly recall in the following.
\textbf{Price of Anarchy.}
%
%
The Price of Anarchy (PoA) of mechanism $\mathbb{M}$ is defined as the maximum ratio between the optimal social welfare and the  welfare in the worst Nash Equilibrium, hence 
\[
PoA(\mathbb{M}):=\sup_{I\in \mathcal{I}}\frac{w(\mu(I))}{w(\mu_{wNE}(I))},
\]
where $\mu_{wNE}(I)$ is the output of $\mathbb{M}$ when the agents act according to the worst Nash Equilibrium, i.e. the Nash Equilibrium that achieves the worst social welfare.

\textbf{Price of Stability.}
The Price of Stability (PoS) of a mechanism $\mathbb{M}$ is defined as the maximum ratio between the optimal social welfare and the welfare in the best Nash Equilibrium, hence 
\[
PoS(\mathbb{M}):=\sup_{I\in \mathcal{I}}\frac{w(\mu(I))}{w(\mu_{bNE}(I))},
\]
where $\mu_{bNE}(I)$ is the output of $\mathbb{M}$ when the agents act according to the best Nash Equilibrium, i.e. the Nash Equilibrium that achieves the maximum social welfare.
Notice that, by definition, we have $PoS(\mathbb{M})\le PoA(\mathbb{M})$.

\textbf{Approximation Ratio.} 
The approximation ratio of a truthful mechanism $\mathbb{M}$ is defined as the maximum ratio between the optimal social welfare and the welfare returned by $\mathbb{M}$.
Hence, we have 
\[
ar(\mathbb{M}):=\sup_{I\in \mathcal{I}}\frac{w(\mu(I))}{w(\mu_{\mathbb{M}}(I))},
\]
where $\mu_{\mathbb{M}}(I)$ is the output of $\mathbb{M}$ when $I$ is given in input.
%


\section{The Truthfulness of the Mechanisms}
\label{sec:game_theoretical_setting}

In this section, we study the truthfulness of the three mechanisms induced by Algorithm \ref{alg:Maxvalue} and its approximation version. 
We show that, although $\MB$ and $\MD$ are optimal, they are not truthful due to an impossibility result.
Furthermore, we show that the manipulability of a mechanism is related to the length of the augmenting paths found during the routine of Algorithm \ref{alg:Maxvalue} and use this characterization to prove that $\MG$ is truthful.

\begin{thm}
\label{th:no_strategyproof_agent}
There is no deterministic truthful mechanism that always returns an MVbM.
\end{thm}

\begin{proof}
We show this using a counterexample.
Consider two agents $a_1$ and $a_2$ and three tasks $t_1$, $t_2$, $t_3$.
The edge set is $E=\{(a_1,t_1), (a_1,t_2),(a_2,t_1), (a_2,t_3)\}$. 
The values of the three tasks are $q_1 = 1, q_2 = 0.1$, and $q_3=0.1$, respectively, while the capacity of both agents is $b_1=b_2=1$. 
It is easy to see that the optimal matching is not unique. 
In particular, both $\{(a_1,t_1)$, $(a_2,t_3)\}$ and $\{(a_1,t_2),(a_2,t_1)\}$ are feasible solutions whose total weight is $1.1$.
Let us assume that the mechanism returns $\{(a_1,t_1)$, $(a_2,t_3)\}$. 
In this case, if agent $a_2$ does not report edge $(a_2,t_3)$, the maximum matching becomes $\{(a_1,t_2),  (a_2,t_1)\}$.
According to this, agent $a_2$'s utility increases from $0.1$ to $1$.
Similarly, if the mechanism returns $\{(a_1,t_2),(a_2,t_1)\}$, agent $a_1$ can manipulate the result by hiding the edge $(a_1,t_2)$.
Therefore, there is no deterministic truthful mechanism that always returns a maximum matching.
\end{proof}

From Theorem \ref{th:no_strategyproof_agent}, we infer that both $\MB$ and $\MD$ are not truthful and, thus, not group strategyproof.
In the following, we characterize a sufficient condition under which agents' best strategy is to report truthfully.
This characterization will allow us to deduce that $\MG$ is truthful and to present sets of instances on which $\MB$ and $\MD$ are truthful.

\begin{lemma}
\label{thm:truthfulness_cases}
Let us consider an instance in which Algorithm \ref{alg:Maxvalue} completes its routine using only augmenting paths of length equal to $1$. 
If an agent cannot misreport edges in such a way that the algorithm will find an augmenting path of length greater than $1$ in its implementation, then the agent's best strategy is to report truthfully.
\end{lemma}

\begin{proof}

%
Let us denote with $\mu_k$ the matching found at the $k$-th step.
It is easy to see that Algorithm \ref{alg:Maxvalue} concludes its routine using only augmenting paths of length equal to $1$ if and only if the sequence of matching that it finds is monotone increasing. That is, $\mu_k\subset \mu_{k+1}$ for every $k$.
Since the sequence $\{\mu_k\}_k$ is increasing, the matching $\mu_{k+1}$ is defined by adding (at most) an edge to $\mu_k$. 
%
Let us now assume that an agent hides a set of edges in such a way that the matching sequence found by the algorithm is still monotone increasing.
Thus, after the manipulation, the final matching is still obtained by adding (at most) an edge at every iteration. 
By the definition of $\MB$ and $\MD$, the edge added at step $j$, is taken from the ones that are connected to the task $t_j$.
Therefore, by hiding one or more edges, the manipulative agent can only reduce the total value of the tasks assigned to it, which concludes the proof. 
\end{proof}

Since $\MG$ uses only augmenting paths of length at most equal to $1$, regardless of the input, Lemma \ref{thm:truthfulness_cases} allows us to conclude that $\MG$ is truthful.
Moreover, due to the results proven in \cite{dobrian20192}, we can characterize the approximation ration of $\MG$.

\begin{thm}
\label{thm:approx_truth}
The mechanism $\MG$ is truthful with respect to agent manipulation. 
Moreover, its approximation ratio is $2$.
\end{thm}

The approximation ratio achieved by $\MG$ is actually the best possible ratio achievable, as the next result shows.

\begin{thm}
\label{thm:tight_bound_general}
  There does not exist a truthful and deterministic mechanism for the MVbM problem that achieves an approximation ratio better than $2$.  
\end{thm}

\begin{proof}
    Toward a contradiction, assume that there exists a mechanism $\mathbb{M}$ whose approximation ratio is equal to $2-\delta$, where $\delta$ is a positive value.
    Let us now consider the following instance.
    We have two agents, namely $a_1$ and $a_2$, and two tasks, namely $t_1$ and $t_2$.
    The capacity of both agents is set to be equal to $1$.
    The value $q_1$ of $t_1$ is $1+\epsilon$ and the value $q_2$ of $t_2$ is equal to $1$.
    Finally, we assume that, according to their truthful inputs, both agents are connected to both tasks, so that $E=\{(a_1,t_1),(a_1,t_2),(a_2,t_1),(a_2,t_2)\}$.
    It is easy to see that the maximum value that a matching can achieve is $2+\epsilon$. 
    If the mechanism $\mathbb{M}$ does not allocate both the tasks, we have that the value achieved by the mechanism is, at most $1+\epsilon$.
    Therefore, we have that the approximated ratio of $\mathbb{M}$ is at least equal to $\frac{2+\epsilon}{1+\epsilon}=1+\frac{1}{1+\epsilon}$.
    If we take $\epsilon <\frac{\delta}{1-\delta}$, we get that the approximated ratio of $\mathbb{M}$ should be greater than $2-\delta$, which is a contradiction.
    Hence $\mathbb{M}$ allocates both the tasks in the previously described instance.
    Without loss of generality, let us assume that $\mathbb{M}$ allocates $t_1$ to $a_2$ and $t_2$ to $a_1$.
    Let us now consider the instance in which $a_1$ is not connected with the task $t_2$, so that $E'=\{(a_1,t_1),(a_2,t_1),(a_2,t_2)\}$.
    The maximal value that a matching can achieve is still $2+\epsilon$.
    However, since the mechanism is truthful, the agent $a_1$ cannot receive any task.
    Indeed, the only task that $\mathbb{M}$ can assign to $a_1$ is $t_1$, however, if $\mathbb{M}$ assigns $t_1$ to $a_1$, it means that agent $a_1$ can manipulate $\mathbb{M}$ by reporting the set of edges $\{(a_1,t_1)\}$ over $\{(a_1,t_1),(a_1,t_2)\}$ in the instance when the truthful input is $E$.
    Since the first agent cannot receive any tasks from $\mathbb{M}$, we have that the maximum matching value achieved by $\mathbb{M}$ when the input is $E'$ is, at most, $1+\epsilon$, so that the approximation ratio of $\mathbb{M}$ is, at least $\frac{2+\epsilon}{1+\epsilon}$.
    Again, by taking $\epsilon<\frac{\delta}{1-\delta}$, we conclude a contradiction.
    \end{proof}

From Theorem \ref{thm:tight_bound_general}, we then infer that $\MG$ is the best possible deterministic and truthful mechanism for our game theoretical setting.
To conclude, we show that, $\MG$ is also  group strategyproof if all the tasks have different values.

\begin{thm}
\label{thm:strategyproof}
    If all the tasks in $T$ have different values, then $\MG$ is group strategyproof.
\end{thm}

\begin{proof}
    Toward a contradiction, let us assume that there exists a coalition of agents $C=\{a_{i_1},\dots,a_{i_\ell}\}$ that is able to collude.
    Since hiding edges that are not returned by $\MG$ does not alter the outcome of the mechanism, we assume, without loss of generality, that at least agent $a_{i_1}$, hides one of the edges that are in the matching found by $\MG$ when all the agents report truthfully.
    Let us denote with $t_l$ the task connected to $a_{i_1}$ through the hidden edge.
    After misreporting agent $a_{i_1}$ cannot be allocated with $t_l$.
    Furthermore, due to the routine of $\MG$, each task is allocated independently from the others, hence $a_{i_1}$ cannot be allocated with a better task.
    Finally, since there are no tasks with the same value, $a_{i_1}$'s payoff is necessarily lowered by misreporting, even if in a coalition, which is a contradiction.
    %
    %
\end{proof}

The condition of Theorem \ref{thm:strategyproof} are tight.
Indeed, as we show in the next example, even if just two tasks have the same value, the mechanism is no longer group strategyproof.

\begin{example}
\label{ex:collusion}
Let us consider the following instance. 
The set of agents is composed of three elements, namely $a_1$, $a_2$, and $a_3$.
The capacity of each agent is set to $1$, so that $b_1=b_2=b_3=1$.
The set of tasks is composed of two elements, namely $t_1$ and $t_2$.
The value of both tasks is equal to $1$.
Finally, let $E=\{(a_1,t_1),(a_1,t_2),(a_2,t_2),(a_3,t_1)\}$ be the truthful input.
Then, $\MG(E)=\{(a_1,t_1),(a_2,t_2)\}$.
However $a_1$ and $a_3$ can collude: if agent $a_1$ hides the edge $(a_1,t_1)$, the $\MG$ returns $\mu'=\{(a_1,t_2),(a_3,t_1)\}$. 
\end{example}

\section{The Price of Anarchy and the Price of Stability}

\label{sect:strategy}

In this section, we study to what extent an agent can manipulate $\MB$ or $\MD$ in its favor. 
First, we show that an agent who is not matched to any task when reporting truthfully cannot improve its utility by manipulation.

\begin{lemma}
\label{lemma:nopayoff}
Given a truthful input, if an agent receives a null utility from $\MB$, then its utility cannot be improved by hiding edges.
The same holds for the mechanism $\MD$.
\end{lemma}



%
The latter Lemma ensures us that, without loss of generality, if an agent is able to manipulate $\MB$ or $\MD$, then the same agent is assigned at least one task when it reports truthfully.
We now show that the first agent processed by either $\MB$ or $\MD$, hence $a_1$, is always able to get its highest possible payoff by misreporting.

\begin{thm}
\label{thm:max_match_agent_one}
For both $\MB$ and $\MD$, agent $a_1$'s best strategy is to report only the top $b_1$-valued tasks to which it is connected.
\end{thm}

\begin{proof}
%
Let us denote with $t_{j_1},\dots,t_{j_{b_1}}$ the top $b_1$ tasks agent $a_1$ is connected to.
For every $t_{j_r}$, at the $j_r$-th step of Algorithm \ref{alg:Maxvalue}, agent $a_1$ will not be saturated; therefore the path $P_{j_r}=\{(t_{j_r},a_1)\}$ is augmenting with respect to matching $\mu_{j_r-1}$.
Moreover, since the BFS searches among the vertices in lexicographical order, the path $P_{j_r}$ is always the first one being explored and, since it is augmenting, it will be the one returned. 
To conclude, we notice that, after the $j_r$-th iteration, there are no augmenting paths that pass by $a_1$ as all the edges connected to $a_1$ are already in the matching, thus the set of tasks allocated to $a_1$ will not change in later iterations of the algorithm.
By a similar argument, we infer the same conclusion for $\MD$.
\end{proof}

In the next example, we show that the advantage described in Theorem \ref{thm:max_match_agent_one} is only due to the fact that agent $a_1$ is self-aware of its position in the processing process.
%

\begin{example}
\label{eq:social_welfarealter}
Let us consider the following instance.
There are $3$ agents, namely $\alpha$, $\beta$, and $\gamma$ whose capacities are $b_\alpha=2$ and $b_\beta=b_\gamma=1$. 
Let us consider a set of $4$ tasks, namely $t_j$ with $j\in [4]$ whose respective values are $q_j=2^{-j}$.
In the truthful input, agent $\alpha$ is connected to all $4$ tasks, while agent $\beta$ is connected only to task $t_1$ and agent $\gamma$ is connected only to task $t_2$.
Let us consider the mechanism $\MB$: the maximum matching for the truthful input allocates $t_1$ to agent $\beta$, $t_2$ to agent $\gamma$, and the other two tasks to agent $\alpha$. 
Let us now assume that the processing order of the agents is $\alpha$, $\beta$, and $\gamma$. That is, $a_1=\alpha$, $a_2=\beta$, and $a_3=\gamma$.
Then, if agent $\alpha$ applies the strategy highlighted in Theorem \ref{thm:max_match_agent_one} and reports only the first two edges, it improves its utility from $q_3+q_4$ to $q_1+q_2$.
However, in a different order of agents in which agent $\alpha$ is the second, i.e., $a_2=\alpha$, if it applies the same strategy, then agent $\alpha$ gets only one of the two tasks (depending on who is the agent going first) while if it goes as the third, it receives no tasks.
We also note that, if agent $\alpha$ is the second, its best strategy is to report the edges connecting it to tasks $t_1$ and $t_3$ if agent $\gamma$ goes first or tasks $t_2$ and $t_3$ if $\beta$ goes first.
Hence, the priority of agent $\alpha$ and the priority of the other two agents determines what the best strategy for agent $\alpha$ is.
Similarly, the same conclusion can be drawn for the mechanism $\MD$.
\end{example}

%
%

%
In general, the best strategy of an agent depends on its priority and the reports of the agents whose priority is higher than its.
Indeed, once the agents' order is fixed, it is possible to describe the Nash Equilibria of both $\MB$ and $\MD$.
For every $i=0,1,\dots,n$, let us define the sets $T^{(i)}$ and $B^{(i)}$ in the following iterative way:
 \begin{enumerate*}[label=(\roman*)]
    \item $T^{(0)}=T$ where $T$ is the set of all the tasks and $B^{(0)}=\emptyset$;
    \item $T^{(i)}=T^{(i-1)}\backslash B^{(i-1)}_{i}$, where $B^{(i-1)}_{i}$ is the set containing the top $b_i$-valued tasks among the ones in $T^{(i-1)}$ that agent $i$ is connected to. 
    If agent $i$ is connected to less than $b_i$ tasks, then $B^{(i-1)}_{i}$ contains all the tasks in $T^{(i-1)}$ to which agent $i$ is connected to.
\end{enumerate*}
We then define agent $a_i$'s \emph{First-Come-First-Served} (FCFS) policy as $\SP_i=\{(a_i,t_j)\}_{t_j\in B^{(i-1)}_i}$.
Notice that $\SP_i\in \mathcal{S}_i$ for every $i\in [n]$, so that it is a feasible strategy for every agent.

\begin{thm}
\label{thm:Nash_Equlibrium}
Given an MVbM problem, the \safepolicies\ constitute a Nash Equilibrium that achieves the lowest social welfare for both mechanisms $\MB$ and $\MD$.
Moreover, we have 
\begin{align}
\label{eq:input=output}
\nonumber \MB(\cup_{a_i\in A}\SP_i)&=\MD(\cup_{a_i\in A}\SP_i)\\&=\cup_{a_i\in A}\SP_i.
\end{align}
\end{thm}

\begin{proof}
We prove the first part of the theorem in two steps.
First, we prove that the \safepolicies\ constitute a Nash Equilibrium.
Second, we show that the Nash Equilibrium they form is the one with the lowest possible social welfare.
Let us then consider agent $a_i$ and, toward a contradiction, let us assume that, when all the other agents report their FCFS policy, reporting the set of edges $S_{a_i}\neq \SP_{i}$ gives $a_i$ a bigger payoff than what it would get from reporting $\SP_i$. 
Let us set $s_i=\min_{(a_i,t_j)\in \SP_{i}} q_j$.
Let us assume that $S_{a_i}$ contains an edge that connects $a_i$ to a task that has a higher value than $s_i$, namely $t_l$.
We now show that, since the other agents are applying their \safepoliciesp, the task $t_l$ is allocated to an agent with higher priority unless the edge already belongs to $\SP_i$. 
Indeed, since $|\SP_j|\le b_j$ for every $a_j\in A$, there cannot be augmenting paths that pass by any of the agents playing their FCFS policy.
Indeed, since the union of the FCFS policies is a $b$-matching, either $(a_i,t_l)\in \SP_i$ or there exists another agent whose FCFS policy connects it to $t_l$.
If $t_l$ is connected to an agent $a_k$, $(a_k,t_l)\in FCFSP_k$, and agent $a_k$'s priority is higher than agent $a_i$'s priority, the final output of the mechanism assigns $t_l$ to $a_k$.
We can then assume that $S_{a_i}$ does not contain edges that connect agent $a_i$ to tasks with values higher than $s_i$ and that do not belong to $\SP_i$.
To conclude, we notice that if $S_{a_i}$ contains an edge connecting agent $a_i$ to a task that has a value lower than $s_i$, then the payoff of agent $a_i$ can only be lowered. 
Indeed, if $|\SP_i|=b_i$, then $a_i$ cannot improve its payoff by reporting edges that connect $a_i$ to tasks that have a value lower than $s_i$.
This follows from the fact that all the other players are using their \safepolicies\ and therefore agent $a_i$ is allocated the set $B^{(i-1)}_i=\{t_j\in T \;\text{s.t.}\; (a_i,t_j)\in \SP_i\}$ if it uses its \safepolicyp.
If $|\SP_i|<b_i$, by definition, it means that there are no tasks that agent $a_i$ can be connected to and that have a value lower than $s_i$.
Therefore $S_{a_i}$ does not contain edges connecting $a_i$ to a task with a value lower than $s_i$ nor edges connecting it with tasks that have a value greater than $s_i$ and that are not included in $\SP_i$.
Since reporting a subset of $\SP_i$ would result in a lower payoff, we deduce that $S_{a_i}=\SP_i$, which is a contradiction since we assumed $\SP_{i}\neq S_{a_i}$.
We now prove that the Nash Equilibrium induced by the \safepolicies\ is one of the worst equilibria.
Toward a contradiction, let us consider another set of strategies, namely $\{S_{a_i}\}_{a_i\in A}$, such that the social welfare achieved by this equilibrium is strictly lower than the one obtained if every agent uses its \safepolicy.
By the definition of social welfare, there must exist at least one agent that, according to the equilibrium defined by $\{S_{a_i}\}_{a_i\in A}$, receives a payoff that is strictly lower than the one it would obtain by using its \safepolicyp.
Let us denote with $a_k$ the first agent that, according to the priority order of the mechanism, receives a lower value.
Agent $a_k$ cannot be the first agent, as it otherwise could apply its \safepolicy\ and get a better payoff.
Then, agent $a_k$ is among the remaining agents and it is not getting any of the tasks that are given to the first agent according to its \safepolicyp, as otherwise, agent $a_1$ could increase its payoff by manipulating and the set of strategies $\{S_i\}_{a_i\in A}$ would not be a Nash Equilibrium.
From a similar argument, we infer that agent $a_k$ cannot be the second agent, as otherwise, it could get a better payoff by using its \safepolicyp.
Moreover, the second agent is allocated the tasks that are granted to it from its \safepolicyp.
Both of which we have already proved cannot be.
By applying the same argument to the other agents, we get a contradiction, since no agent can be agent $a_k$.
We, therefore, conclude that the set of strategies given by the \safepolicies\ is one of the worst Nash Equilibrium.
The last part of the Theorem follows from the fact that $\cup_{a_i\in A}\SP_i$ is itself a $b$-matching, hence it is also an $MVbM$ with respect to the edge set $E=\cup_{a_i\in A}\SP_i$. 
\end{proof}

Following the same argument presented in the proof of Theorem \ref{thm:Nash_Equlibrium}, we infer that any set of strategies $\{S_i\}_{a_i\in A}$ for which it holds $\SP_i\subset S_i$ and $\min_{t_j, \;(a_i,t_j)\in \SP_i} q_j=\min_{t_j, \; (a_i,t_j)\in S_i}q_j$ for every $i$, defines a Nash Equilibrium.
Among this class of Nash Equilibria, $\{\SP_i\}_{a_i\in A}$ is the only one for which the identities in \eqref{eq:input=output} hold. 
%
Moreover, this equilibrium achieves the same social welfare of the matching returned by $\MG$.

\begin{thm}
\label{thm:NE_and_approx}
Given an MVbM problem, let $\mu_E$ be the matching returned by $\MG$.
Then, it holds that $\mu_E=\cup_{a_i\in A}\SP_i$. 
That is, the output of $\MG$ on any given instance is equal to the union of the agents' \safepoliciesp.
Hence, the social welfare achieved by $\MG$ is 
equal to the social welfare achieved by $\MB$ and $\MD$ in one of their worst Nash Equilibria.
\end{thm}

\begin{proof}
We prove this Theorem by induction.
First, we prove that $\MG$ allocates $a_1$ with the set of tasks $\SP_1$.
Second, we show that if all the agents $a_1,\dots,a_k$ receive $\SP_1,\dots,\SP_k$ from $\MG$, then also agent $a_{k+1}$ receives $\SP_{k+1}$.
Let us consider $a_1$. 
Since $\MG$ checks the agents following their orders, $a_1$ is always the first agent checked.
Hence, a task $t_j$ is not allocated to $a_1$ if and only if there are $b_1$ other tasks that have a higher value than $t_j$ and that are all connected to $a_1$.
This means that the set of tasks allocated to $a_1$ is $\SP_1$.
Let us now assume that agents $a_1,a_2,\dots,a_k$ are given the tasks contained in $\SP_1,\SP_2,\dots,\SP_k$ respectively, and let us consider $a_{k+1}$.
We have that $\MG$ allocates a task $t_j$ to $a_{k+1}$ if, at step $j$ of Algorithm \ref{alg:Maxvalue}, all the agents with priorities higher than $a_{k+1}$ that are connected to $t_j$ are already saturated and $a_{k+1}$ is not saturated.
However, by assumption, agents with a higher priority than $a_{k+1}$ are getting the tasks that they would get from their \safepoliciesp.
We, therefore, conclude that the tasks allocated to agent $a_{k+1}$ from $\MG$ consist of a subset of $T^{(k)}$.
From an argument similar to the one used for agent $a_1$, we infer that $a_{k+1}$ receives the top $b_{k+1}$ higher valued tasks among the ones in $T^{(k)}$, which coincides with the set $\SP_{k+1}$.
We conclude the proof by induction.
\end{proof}

Finally, we show that Theorem \ref{thm:NE_and_approx} along with Theorem \ref{thm:approx_truth} allows us to compute the PoA of both $\MB$ and $\MD$.

\begin{thm}
\label{thm:PoA}
The PoA of $\MB$ and $\MD$ is equal to $2$.
\end{thm}

\begin{proof}
From Theorem \ref{thm:NE_and_approx}, the matching returned by $\MG$ achieves a social welfare equal to the social welfare of one of the worst Nash Equilibrium of $\MB$.
Since $\MB$ returns an MVbM, we have
\begin{equation}
    \label{eq:app_socialcompare}
    {\displaystyle PoA(\MB) = \sup_{I\in \mathcal{I}}\frac{w(\mu(I))}{w(\mu_{wNE}(I))}= \sup_{I\in \mathcal{I}} \frac{w(\MB(I))}{w(\MG(I))}.}
\end{equation}
Since the matching found by $\MB$ achieves the maximum social welfare, the last term in equation \eqref{eq:app_socialcompare} is bounded from above by the $a.r.(\MG)$, so that $PoA(\MB) \le ar(\MG)=2$.
To conclude $PoA(\MB)=2$ we show a lower bound of $2$. 
The set of agents is composed of two agents, namely $a_1$ and $a_2$, we assume the agents to be ordered according to the algorithm priority.
The capacity of both agents is equal to $1$. 
The set of tasks contains two tasks, namely $t_1$ and $t_2$, whose values are $1+\epsilon$ and $1$, respectively.
Finally, let us assume that the truthful input is given by $E=\{(a_1,t_1),(a_1,t_2),(a_2,t_1)\}$.
The social welfare is then equal to $2+\epsilon$. 
However, in the worst Nash Equilibrium, the welfare is $1+\epsilon$.
By taking the limit for $\epsilon\to 0$, we conclude that the PoA is equal to $2$.
By a similar argument, we infer $PoA(\MD)=2$.
\end{proof}

The previous bound is tight: there does not exist a deterministic mechanism for the $MVbM$ problem that has a $PoA$ lower than $2$.

\begin{thm}
\label{thm:tight_PoA}
    For every deterministic mechanism $\mathbb{M}$, we have $PoA(\mathbb{M})\ge 2$ with respect to agent manipulations.
\end{thm}

\begin{proof}
    %
    Toward a contradiction, let $\mathbb{M}$ be a deterministic mechanism  whose $PoA$ is lower than $2$.
    Let us consider the following instance. 
    We have two tasks, namely $t_1$ and $t_2$, whose values are $1+\epsilon$ and $1$, respectively.
    We then have two agents, namely $a_1$ and $a_2$ and both have a capacity equal to $1$.
    Let us now consider the instance in which both the agents are only connected to task $t_1$, hence the truthful input is $E=\{(a_1,t_1),(a_2,t_1)\}$.
    Since $PoA(\mathbb{M})$ is finite, we have that $\mathbb{M}$ allocates $t_1$ to one agent.
    Indeed, if no agent receives a task, no one can improve its own payoff by hiding its only edge (we recall that each agent is bounded by its statements).
    %
    Hence, the truthful instance is already a Nash Equilibrium.
    Furthermore, since the social welfare of this Nash Equilibrium is $0$, this is also one of the worst Nash Equilibria, thus we find a contradiction since we assumed that $\mathbb{M}$ has a finite PoA.
    Let us then assume that one agent gets $t_1$.
    Without loss of generality, let us assume that $\mathbb{M}$ allocates $t_1$ to $a_1$, the other case is completely symmetric with respect to the one we are about to present.
    Let us now consider the instance whose truthful input is $E=\{(a_1,t_1),(a_1,t_2),(a_2,t_1)\}$.
    If $\mathbb{M}$ allocates $t_1$ to $a_2$, we have that a Nash Equilibrium is obtained when agent $a_1$ hides arc $(a_1,t_2)$.
    Indeed, if agent $a_1$ hides $(a_1,t_2)$, the input of the mechanism is $E=\{(a_1,t_1),(a_2,t_1)\}$ which gives the first task to $a_1$. 
    Since $a_2$ is bounded by its statements and its only alternative is to report no edges, it has no better strategy to play.
    Similarly, since $a_1$ is getting its best possible payoff, it has no better strategy to play.
    We then conclude that when $a_1$ hides the edge $(a_1,t_2)$, we have a Nash Equilibrium.
    Finally, we observe that, by taking $\epsilon$ small enough, we get a contradiction with the assumption $PoA(\mathbb{M})<2$, since the maximum social welfare is $2+\epsilon$, while the social welfare returned by the mechanism in the worst Nash Equilibrium is at most $1+\epsilon$.
    Notice that there might be another Nash Equilibrium in which the social welfare is even lower, however, it suffice to notice that the social welfare achieved in the worst Nash Equilibrium is lower than $1+\epsilon$ to conclude the proof.
    Similarly, if the mechanism does not allocate $t_1$ to $a_2$, the instance is already a Nash Equilibrium.
    Indeed, by the same argument used before, $a_2$ cannot improve its own payoff, since it is getting no tasks.
    If $a_1$ is allocated with the task, it cannot improve its payoff  either, since it is getting the maximum payoff it can get.
    Finally, if $a_1$ is not getting $t_1$, it can hide $(a_1,t_2)$ and return to the instance we considered before.
    Again, by taking $\epsilon$ small enough, we retrieve that the $PoA(\mathbb{M})$ cannot be less than $2$.
\end{proof}

We close the section by studying the PoS of $\MB$ and $\MD$.
We recall that the PoA of every mechanism is greater than its PoS, thus we infer that both $\MB$ and $\MD$ have a PoS at most equal to $2$.
Moreover, since the Nash Equilibrium in the example we used in the proof of Theorem \ref{thm:tight_PoA} is unique, the best and worst Nash Equilibria achieve the same social welfare. In particular, this allows us to prove that the PoS of both $\MB$ and $\MD$ is equal to $2$.
Furthermore, this value is tight for the class of deterministic mechanisms.

\begin{thm}
\label{thm:PoS}
    The PoS of $\MB$ and $\MD$ is equal to $2$.
    Moreover, no deterministic mechanism can achieve a lower PoS.
\end{thm}




\section{The Truthful Inputs for $\MB$ and $\MD$}
\label{subsec:truthinstances}

In this section, we describe three sets of inputs in which $\MB$ and $\MD$ are truthful.
In particular, we consider the three following settings.
In Theorem \ref{thm:truthful_degree}, we consider the case in which there is a shortage of tasks, so that the capacity of each agent exceeds the number of tasks to which it is connected.
In particular, no agent can be saturated.
In Theorem \ref{thm:truth_BFS}, instead, we describe what happens when every task can be contended by another agent.
%
%
%
%
In Theorem \ref{thm:agent_classes}, we study the case in which the private information of the agents can be clustered, that is different agents are connected to the same set of tasks and the same capacity.
Going back to the worker-project example at the beginning of the paper, this means that the connection between the worker and the project depends, for example, on the field of expertise of the worker or their background formation.
%
%

\begin{thm}
\label{thm:truthful_degree}
Let us consider the set of inputs such that, according to $E$, all the agents' degrees are less than or equal to their capacity, that is $\sum_{t_j\in T_i}e_{i,j}\le b_i$ for every $i\in [n]$, then $\MB$ and $\MD$ are truthful on this set of inputs.
\end{thm}


\begin{proof}
%
Since no agents can be saturated, any augmenting path found by the mechanisms has length equal to $1$.
Moreover, since hiding edges cannot lead an agent to be saturated, Lemma \ref{thm:truthfulness_cases} allows us to conclude the proof.
\end{proof}

This is the only class of inputs we are considering on which $\MB$ and $\MD$ behave in the same way.
In the other two frameworks, $\MD$ is not truthful, while its counterpart $\MB$ is.
This is due to the fact that BFS searches for the shortest possible augmenting path in its execution.

\begin{thm}
\label{thm:truth_BFS}
Let $\mu$ be the matching returned by $\MB$.
If for every task $t_j$ there exists an edge $e\notin \mu$ that connects $t_j$ to an unsaturated agent, then no agents can increase their utility by hiding only one edge. 
In particular, if the truthful input is a complete bipartite graph and the vector of the capacities $\bb$ is such that $m\le \sum_{i=1}^{n} b_i-\max_{i\in [n]}{b_i}$, then the best strategy for every agent is to report truthfully.
\end{thm}

\begin{proof}
Assume, toward a contradiction, that an agent, namely $a_i$ gets a benefit by hiding an edge, namely $e=(a_i,t_j)$.
By hypothesis, we have that every task is connected to an unsaturated agent.
Let $a_k$ be one of the unsaturated agents to which $t_j$ is connected to.
Then BFS will always find an augmenting path whose length is $1$ when it is asked to allocate $t_j$, since there exists the augmenting path $(a_k,t_j)$.
Therefore, by Lemma \ref{thm:truthfulness_cases} we infer a contradiction.
\end{proof}



%
Let $\AA=\{A^{(1)},\dots,A^{(r)}\}$ be a partition of $A$.
We say that $A^{(\ell)}$ is the $\ell$-th class of the agents.
Since $\AA$ is a partition, every agent $a_i\in A$ belongs to only one class.
Let us assume that the capacity $b_i$ and the set of edges $T_i$ of every agent $a_i\in A$ depends only on the class $A^{(\ell)}$ to which $a_i$ belongs, so that $b_i=b^{(\ell)}$ and $T_i=T^{(\ell)}$.
Using an argument that is similar (but more delicate) to the one used in Theorem \ref{thm:truth_BFS}, we are able to prove that $\MB$ is truthful if every class contains enough agents.

\begin{thm}
\label{thm:agent_classes}
In the framework described above, if $|A^{(\ell)}|>\big\lceil \frac{|T^{(\ell)}|}{b^{(\ell)}} \big\rceil+1$, then no agents belonging to the $\ell$-th class can benefit by misreporting to $\MB$.
\end{thm}

In the appendix, we report two examples that show that both Theorem \ref{thm:truth_BFS} and Theorem \ref{thm:agent_classes} do not hold for $\MD$.
%


\section{Agents Manipulating their Edges and Capacity}
\label{subsect:ECMS}

Finally, we extend our study on the truthfulness of $\MB$, $\MD$, and $\MG$ to the ECMS, i.e. we allow the agents self-report their capacity along with their edges.
As for the EMS, we assume the agents to be bounded by their statements, thus they can manipulate only by hiding edges or by reporting a lower capacity than their real one.
In this setting, a mechanism $\mathbb{M}$ is truthful if, for every $i\in [n]$, it holds $w_i((I'_i,b_i'),J_{-i})\le w_i((I_i,b_i),J_{-i})$, for every $(I_i',b'_i)\in\mathcal{S}_i\times [b_i]$, where $J_{-i}$ are the reports of the other agents.
Once we fix the set of strategies of each agent, we can define the PoA, PoS, and approximation ratio of a mechanism $\mathbb{M}$ as for the EMS by carefully changing the set of strategies to fit the ECMS case.
With a slight abuse of notation, we still use $\MB$, $\MD$, and $\MG$ to denote the mechanisms obtained from Algorithm \ref{alg:Maxvalue} and its approximation version.
%
As we show, neither the truthfulness nor the efficiency guarantees of $\MB$, $\MD$, and $\MG$ change from EMS to ECMS. 
Furthermore, all the bounds are still tight.
%
%

%

\begin{thm}
\label{thm:PoA_PoS_edge_cap}
    In the ECMS, $\MB$ and $\MD$ are both not truthful.
    The PoA and the PoS of both $\MB$ and $\MD$ are equal to $2$.
    Moreover, these bounds are tight, hence there is no other deterministic mechanism whose PoA or PoS is lower.
\end{thm}

Similarly, $\MG$ is still truthful in the ECMS, and its approximation ratio is unchanged. 

\begin{thm}
\label{thm:last_mg_truth_cap}
In the ECMS, $\MG$ is truthful, its approximation ratio is equal to $2$.
Moreover, there is no deterministic truthful mechanism with a better approximation ratio.
Finally, if the tasks have different values, $\MG$ is group strategyproof.
\end{thm}


\section{Conclusion and Future Work}
\label{sect:conclusion}

In this paper, we propose a new game-theoretical framework for MVbM 
problems, where one side of the bipartite graph consists of agents and the other side consists of tasks with objective values.
We consider scenarios where agents can behave strategically by hiding connections with tasks or lowering their capacity. We analyze three mechanisms in this framework: $\MB$, $\MD$, and $\MG$.
We first show that these mechanisms are either optimal ($\MB$ and $\MD$) or truthful ($\MG$). Then, we demonstrate that these mechanisms are also the best in terms of PoA (Price of Anarchy), PoS (Price of Stability), and approximation ratio. In other words, no other mechanisms can outperform these ones with respect to these performance measures in our setting.

A future direction of our work is to study the effect of shuffling the agents' order on the manipulability of $\MB$.
Specifically, we are interested in investigating whether randomizing the priority of agents further highlights the differences between $\MB$ and $\MD$.
Moreover, it would be interesting to explore how placing bounds on the number of edges that each agent can report affects the performance and manipulability of the mechanisms we studied.
Finally, we plan to investigate more in details other classic game-theoretical aspects of our model, such as fairness and envy-freeness.



\ack We would like to thank the referees for their comments, which
helped improve this paper considerably.
This project is partially supported by a Leverhulme Trust Research Project Grant (2021 -- 2024). Jie Zhang is also supported by the EPSRC grant (EP/W014912/1).

\bibliography{ecai}

\clearpage

\section*{Appendix}

In this appendix, we report the proofs and examples we omitted from the main body of the paper.

\subsection*{Missing Proofs}

\begin{proof}[Proof of Theorem \ref{thm:approx_truth}]
The first part of the statement follows directly from Lemma \ref{thm:truthfulness_cases}.
The second part of the statement has been proven in \cite{dobrian20192}.
Indeed, in \cite{dobrian20192}, it has been shown that the weight of the matching returned by the approximation algorithm that uses augmenting paths of length equal to $1$ is, in the worst case, half of the maximal weight.
In their paper, Dobrian et al study the case in which both sides of the bipartite graph have a non-negative value.
However, if we set the value of each agent in $A$ to be $0$, we have that the weight of the matching studied in \cite{dobrian20192} is the same as the social welfare (from the agents' viewpoint) we defined in the present paper.
Similarly, the maximum weight of the matching is the maximum social welfare achievable.
We, therefore, conclude that the ratio between the social welfare returned by $\MG$ and the maximum one is greater than $0.5$.
Since the approximation ratio is the ratio between the maximum social welfare and the one achieved by $\MG$, we infer that it has to be smaller than $2$.
To prove that the approximation ratio of the mechanism is $2$, consider the following instances.
There are two agents, namely $a_1$ and $a_2$, and two tasks $t_1$ and $t_2$, whose values are $1+\epsilon$ and $1$, respectively, where $\epsilon>0$.
Let $E=\{(a_1,t_1),(a_1,t_2),(a_2,t_1)\}$ be the truthful set of edges.
It is easy to see that $\MG$ returns the matching $\mu_{\MG}=\{(a_1,t_1)\}$, while the optimal one is $\{(a_2,t_1),(a_1,t_2)\}$.
Hence the ratio between the maximum social welfare and the social welfare achieved by $\MG$ on this instance is $\frac{2+\epsilon}{1+\epsilon}$.
By taking the limit for $\epsilon\to 0$, we conclude the proof.
\end{proof}

\begin{proof}[Proof of Lemma \ref{lemma:nopayoff}]
For the sake of simplicity, we present the proof only for the mechanism $\MB$.
Let us denote with $G$ the truthful graph and with $G'$ the graph manipulated by agent $a_k$. 
Let $\mu$ and $\mu'$ be the matching returned by $\MB$ when $G$ and $G'$ are given as input, respectively. 
By contradiction, suppose $\mu'\neq\mu$. In this case, there exists $t_j$ such that $P_j\neq P_j'$ and $P_r=P_r'$ for any $r<j$, where $P_l$ and $P_l'$ are the augmenting paths returned by the BFS at the $l$-th step.
Since $G'\subset G$ and $\mu_{l-1}=\mu'_{l-1}$, we must have that the last edge of $P_l$ is in $G$ and not in $G'$.
However, since $G'$ is only missing edges connected to $a_k$,  we must have that $P_l$ contains an edge that connects agent $a_k$ to a task. 
In particular, at the $j$-th step, the augmenting path returned by the BFS contains an edge connected to $a_k$, which is not possible since $a_k$ is left unmatched by $\mu$.
\end{proof}

\begin{proof}[Proof of Theorem \ref{thm:PoS}]
    Since the Price of Anarchy of a mechanism is always greater than the Price of Stability, we have that $PoS(\MB)\le 2$.
    We now show that $PoS(\MB)= 2$ by showing an instance on which the PoS is equal to $2$.
    Let us consider the example built in the proof of Theorem \ref{thm:PoA}.
    Since there is only one Nash Equilibrium, it is both the worst and best Nash Equilibrium.
    Therefore, following the argument used in the proof of Theorem \ref{thm:PoA}, we conclude $PoS(\MB)\ge 2$, hence $PoS(\MB)=2$.
    From a similar argument, we infer $PoS(\MD)=2$.
    Finally, to prove that the PoS of every deterministic mechanism is greater than $2$, consider the instance build in the proof of Theorem \ref{thm:tight_PoA}.
    Since the Nash Equilibrium of this instance is unique, we conclude that $PoS(\mathbb{M}_{BFS})\ge 2$ for every deterministic and optimal mechanism.
\end{proof}

\begin{proof}
[Proof of Theorem \ref{thm:agent_classes}]
First, we prove that if there are $k$ agents that report the same edge sets (but that do not have necessarily the same capacity), namely $a_1$, $\dots$, $a_k$, then, if agent $a_l$ is allocated at least a task, all the agents $a_1$,$\dots$, $a_{l-1}$ are saturated.
Moreover, if agent $a_l$ is unsaturated, then all the agents $a_{l+1}$, $\dots$, $a_k$ do not receive any task.
Toward a contradiction, assume that agent $a_k$ is allocated a task and that $a_{k-1}$ is not saturated. 
Since agent $a_k$ receives a task, it means that during an iteration of the Algorithm, the augmenting path ends in $a_k$ with an edge $(t_j,a_k)$. 
However, since $a_{k-1}$ was not saturated and it has the same edges as $a_k$, the BFS should have explored the edge $(t_k,a_{k-1})$ before, which is a contradiction.
Through a similar argument, it is possible to show the second part of the statement.
We are now ready to prove the claim of the Theorem.
Indeed, since every class of agent has at least $\alpha_i+1$ agents, for every class, there is at least one agent that is left unmatched.
Indeed, every agent from the $i$-th class can be allocated at most $b_i$ tasks, hence, if every agent is allocated at least one task, from the statement we proved before, we infer that at least $\alpha_i$ agents are saturated. 
Since we have
\begin{equation}
    \label{eq:app}
    |T_i|\le \alpha_i b_i, 
\end{equation}
we conclude that at least one agent is left unmatched.
We now show that the routine of Algorithm \ref{alg:Maxvalue} over the truthful input terminates without using augmenting paths of length greater than $1$. 
Assume, toward a contradiction, that the routine uses an augmenting path of length greater than $1$ to allocate a task $t$. Since there exists at least one agent that is unmatched and connected to $t$, this cannot be. 
If we show that no agents can force the algorithm to use an augmenting path whose length is greater than $1$, we conclude the proof using Lemma \ref{thm:truthfulness_cases}.
Toward a contradiction, assume that an agent can force an augmenting path of length greater than $1$. 
Thus, during the implementation of $\MB$ of the manipulated input, there exists a task, namely $t_j\in T_i$, that is connected only to saturated agents at step $j$.
Finally, we infer a contradiction from equation \eqref{eq:app}, since there are at least $\alpha_i$ agents in the $i$-th category (the manipulating agent is no longer in this class).
\end{proof}

\begin{proof}[Proof of Theorem \ref{thm:PoA_PoS_edge_cap}]

    First, we observe that, since the strategy set of every agent in the ECMS is larger than what it is in the EMS, the mechanisms cannot be truthful in the ECMS.
    Indeed, if the mechanisms are not truthful for the EMS, they cannot be truthful for the ECMS.
    To prove that $PoA(\MB)=PoA(\MD)=2$ it suffices to show that the strategies $\{(FCFS_i,b_i)\}_{i\in [n]}$ forms a Nash Equilibrium for both $\MB$ and $\MD$.
    Without loss of generality, we focus on $\MB$, since the same conclusion can be drawn for $\MD$.
    Toward a contradiction, let us assume that there exists an agent, namely $a_k$, whose best strategy is to play $(S_k,c_k)$ when all the other agents play $(FCFS_i,b_i)$.
    Since $c_k\le b_k$, agent $a_k$ gets at most the same number of tasks that it would get in the truthful case.
    Furthermore, since the tasks that have a higher value than the one $a_k$ would get by playing $(FCFS_k,b_k)$ cannot be de-allocated from the agents already getting them (since every other player is playing $(FCFS_i,b_i)$), agent $a_k$ is getting fewer tasks that have a lower value than the ones it would get by playing $(FCFS_k,b_k)$, which is a contradiction.
    We deduce that this Nash Equilibrium is one of the worst ones, by using the same argument used in the proof of Theorem \ref{thm:Nash_Equlibrium}.
    Indeed, there cannot exist a Nash Equilibrium in which one agent gets a payoff lower than the one it would get by playing $(FCFS_i,b_i)$.
    We therefore conclude that $\{(FCFS_i,b_i)\}_{i\in [n]}$ is one of the worst Nash Equilibria, thus the PoA of $\MB$ (and $\MD$) is equal to $2$.
    To conclude that $PoS(\MB)=PoS(\MD)=2$, consider the instance described in the proof of Theorem \ref{thm:PoS}.
    Indeed, since in this instance, there is only one possible Nash Equilibrium, it is both the worst and best equilibrium possible.
    Moreover, since the PoS of this instance is equal to $2$ and $PoS(\MB)\le PoA(\MB)=2$, we conclude the proof.
    Similarly, we have $PoS(\MD)=2$.
    To show that all the bounds are tight, it suffice to notice that, according to our framework, an agent whose capacity is equal to $1$ cannot manipulate the mechanism by reporting a null capacity.
    Indeed, an agent that reports a null capacity would automatically be excluded from the allocation procedure.
    Since in the instances we used to prove Theorem \ref{thm:tight_PoA} and \ref{thm:PoS} all the agents have a capacity equal to $1$, we can use them to prove the tightness of the bounds in the ECMS.
\end{proof}

\begin{proof}[Proof of Theorem \ref{thm:last_mg_truth_cap}]
To show the truthfulness, it suffices to prove that if an agent misreports only its capacity, it cannot get a higher payoff.
Indeed, if for any given set of edges reported by an agent its best strategy is to report the maximum capacity, we conclude that its best strategy, in the ECMS, is to report all the edges and the maximum capacity.
Since every agent is bounded by its statements, this implies that $\MG$ is truthful.
Let us assume that the capacity of agent $a_i$ is $b_i$ and that it gets a higher payoff by reporting $b_i'<b_i$ as its capacity.
Let us denote with $\mu_i$ the set of tasks allocated to $a_i$ from $\MG$ when it reports truthfully and let us denote with $\mu_i'$ the set of tasks allocated to $a_i$ when it reports $b_i'$ as its capacity.
Toward a contradiction, let $t_j\in \mu_i'$ such that $t_j\notin \mu_i$.
Since $t_j$ is allocated to agent $a_i$ by $\MG$, it means that agent $a_i$ is unsaturated at the $j$-th iteration of Algorithm \ref{alg:Maxvalue} and since $b_i'<b_i$, we conclude that agent $a_i$ is unsaturated at the $j$-th step also when it reports $b_i$ as its capacity.
Therefore, if task $t_j$ is allocated to agent $a_i$ when it reports $b_i'$ as its capacity, the same holds when $a_i$ reports a capacity of $b_i$.
Since every task has a positive value and $\mu'_i\subset \mu_i$, we get a contradiction, therefore the mechanism cannot be manipulated by only misreporting the capacity.
Since the mechanism $\MG$ cannot be manipulated by hiding edges or by reporting a lower capacity, we conclude that $\MG$ has to be truthful when the agents report both quantities.
Finally, since $\MG$ is truthful in the ECMS and the agents are bounded by their statements, its approximation ratio is the same as the one of $\MG$ in EMS (since there is no difference in having the capacities publicly known and reporting them), hence $ar(\MG)=2$.
Again, to prove that the approximation ratio guarantee is tight, it suffice to consider the instance used to prove Theorem \ref{thm:tight_bound_general}.
The group strategyproofness follows by the same argument used in the proof of Theorem \ref{thm:strategyproof} and by observing that reporting a lower capacity cannot increase the payoff of the agents.
\end{proof}

\subsection*{Missing Examples}

\begin{example}
\label{ex:bip_DFS_vs_BFS}
Let us consider the following instance.
We have three agents, namely $a_i$ for $i=1,2,3$, ordered according to their indexes.
The capacity of each agent is equal to one. 
Let the set of tasks be composed of two tasks $t_1$ and $t_2$, whose values are $1$ and $0.5$, respectively. 
Assume that the edge set of the truthful graph is the complete bipartite one, that is $E:=\{ (a_i,t_j) \}_{i=1,2,3;j=1,2}$. 
For this input $\MB$ returns $\mu_{BFS}=\{(a_1,t_1),(a_2,t_2)\}$, while $\MD$ returns $\mu_{DFS}=\{(a_1,t_2),(a_2,t_1)\}$.
Since agent $1$ does not receive its best pick, by Theorem \ref{thm:max_match_agent_one}, we conclude that the DFS algorithm is susceptible to be manipulated by agent $1$ while, according to Theorem \ref{thm:truth_BFS}, the BFS counterpart is not manipulable by any agent. 
Indeed, Lemma \ref{lemma:nopayoff} ensures us that agent $3$ cannot manipulate $\MB$, since the algorithm leaves it unmatched. 
Similarly, agent $1$ cannot manipulate $\MB$ since it is already receiving its best possible allocation. 
Finally, also $a_2$ cannot manipulate $\MB$, since hiding the only edge used during the implementation of Algorithm \ref{alg:Maxvalue} endowed with the BFS would result in either assigning task $t_2$ to agent $3$ or would not change the allocation at all.
\end{example}

\begin{example}
\label{ex:class_BFS_vs_DFS}
Let us consider the following instance.
We have three tasks, namely $t_j$ with $j=1,2,3$ and $5$ agents, namely $a_i$ with $i=1,\dots, 5$. 
Finally, let us assume that the value of the task $t_j$ is equal to $q_j=3^{-j}$.
Furthermore, let us assume that there are two classes of agents.
The first class has a capacity equal to $2$ and it is connected to all three tasks. 
The second class of agents also has a capacity equal to $2$ and it is connected to only the last $2$ tasks, which are $t_2$ and $t_3$.
Let us now assume that the agents $a_1$, $a_3$, and $a_4$ belong to the first category, while $a_2$ and $a_5$ belong to the second one.
It is easy to see that this problem satisfies the conditions of Theorem \ref{thm:agent_classes}.
We, therefore, deduce that no agent can manipulate $\MB$.
It is easy to see that the output of $\MB$ is $\mu_{BFS}=\{(a_1,t_1),(a_1,t_2),(a_2,t_3)\}$ while $\MD$ returns $\mu_{DFS}=\{(a_2,t_1),(a_1,t_2),(a_1,t_3)\}$.
Since agent $1$ does not receive the top task to which it is connected, we conclude that agent one can manipulate $\MD$ in its favor. 
On the contrary, we notice that no agent can manipulate $\MB$ since $3$ agents receive no payoff, hence they cannot manipulate and the first agent gets its best possible payoff, hence it cannot improve it.
Finally, we notice that if agent $a_2$ manipulates, the only change that it could achieve is that task $t_3$ is allocated to agent $a_3$ rather than to $a_2$.
\end{example}

\end{document}